\def\UseBibLatex{1}

\ifx\ESAVer\undefined%
\newcommand{\InESAVer}[1]{}%
\newcommand{\InRegVer}[1]{#1}%
\else
\newcommand{\InESAVer}[1]{#1}%
\newcommand{\InRegVer}[1]{}%
\fi

\InRegVer{%
   \documentclass[12pt]{article}%
}
\InESAVer{%
   \documentclass[11pt]{article}%
}

\InESAVer{%
   \def\ProofInAppendix{TRUE}%
}

\usepackage{proofapnd}

\providecommand{\BibLatexMode}[1]{}
\providecommand{\BibTexMode}[1]{#1}

\ifx\UseBibLatex\undefined%
  \renewcommand{\BibLatexMode}[1]{}
  \renewcommand{\BibTexMode}[1]{#1}
\else
  \renewcommand{\BibLatexMode}[1]{#1}
  \renewcommand{\BibTexMode}[1]{}
\fi

\BibLatexMode{%
   \usepackage[style=alphabetic,backend=biber,bibencoding=utf8]{biblatex}%
   \usepackage{sariel_biblatex}%
}

\InESAVer{%
   \newcommand{\MedSkip}{}%
}
\InRegVer{%
   \newcommand{\MedSkip}{\medskip}%
}

\usepackage[english]{babel}
\usepackage{euscript}
\usepackage{graphicx}

\IfFileExists{sariel_computer.sty}{\def\sarielComp{1}}{}
\ifx\sarielComp\undefined%
\newcommand{\SarielComp}[1]{}
\newcommand{\NotSarielComp}[1]{#1}%
\else
\newcommand{\SarielComp}[1]{#1}%
\newcommand{\NotSarielComp}[1]{}%
\fi
\newcommand{\IfPrinterVer}[2]{#2}%

\InESAVer{%
}

\InESAVer{%
   \usepackage{fullpage}%
}
\InRegVer{%
   \usepackage[cm]{fullpage}%
}

\usepackage{amsmath}%
\usepackage{amssymb}%
\usepackage{xcolor}%

\SarielComp{\usepackage{sariel_colors}}%

\usepackage[amsmath,thmmarks]{ntheorem}%
\theoremseparator{.}%

\usepackage{titlesec}%
\titlelabel{\thetitle. }%
\usepackage{xcolor}%
\usepackage{mleftright}%
\usepackage{xspace}%
\usepackage{hyperref}%
\usepackage{stmaryrd}%
\usepackage{multirow}%

\IfPrinterVer{%
   \usepackage{hyperref}%
}{%
   \usepackage{hyperref}%
   \hypersetup{%
      breaklinks,%
      ocgcolorlinks, colorlinks=true,%
      urlcolor=[rgb]{0.25,0.0,0.0},%
      linkcolor=[rgb]{0.5,0.0,0.0},%
      citecolor=[rgb]{0,0.2,0.445},%
      filecolor=[rgb]{0,0,0.4},
      anchorcolor=[rgb]={0.0,0.1,0.2}%
   }
}

\theoremseparator{.}%

\theoremstyle{plain}%
\newtheorem{theorem}{Theorem}[section]

\newtheorem{lemma}[theorem]{Lemma}

\newtheorem{corollary}[theorem]{Corollary}

\newtheorem{observation}[theorem]{Observation}
\newtheorem*{observation:u}[FakeCounter]{Observation}

\theoremstyle{plain}%
\theoremheaderfont{\sf} \theorembodyfont{\upshape}%
\newtheorem*{remark:unnumbered}[FakeCounter]{Remark}%
\newtheorem{definition}[theorem]{Definition}
\newtheorem{defn}[theorem]{Definition}

\newcommand{\myqedsymbol}{\rule{2mm}{2mm}}

\theoremheaderfont{\em}%
\theorembodyfont{\upshape}%
\theoremstyle{nonumberplain}%
\theoremseparator{}%
\theoremsymbol{\myqedsymbol}%
\newtheorem{proof}{Proof:}%

\newcommand{\atgen}{\symbol{'100}}%
\newcommand{\SarielThanks}[1]{%
   \thanks{%
      Department of Computer Science; %
      University of Illinois; 201 N. Goodwin Avenue; Urbana, IL,
      61801, USA; {\tt sariel\atgen{}illinois.edu}; {\tt
         \url{http://sarielhp.org/}.} #1}}

\newcommand{\JoeThanks}[1]{%
   \thanks{%
      Department of Computer Science; %
      University of Illinois; %
      201 N. Goodwin Avenue; Urbana, IL, 61801, USA; %
      {\tt jwrogge2\atgen{}illinois.edu}; %
      #1}}

\newcommand{\HLink}[2]{\hyperref[#2]{#1~\ref*{#2}}}
\newcommand{\HLinkSuffix}[3]{\hyperref[#2]{#1\ref*{#2}{#3}}}

\newcommand{\figlab}[1]{\label{fig:#1}}
\newcommand{\figref}[1]{\HLink{Figure}{fig:#1}}

\newcommand{\thmlab}[1]{{\label{theo:#1}}}
\newcommand{\thmref}[1]{\HLink{Theorem}{theo:#1}}

\newcommand{\corlab}[1]{\label{cor:#1}}
\newcommand{\corref}[1]{\HLink{Corollary}{cor:#1}}%

\newcommand{\obslab}[1]{\label{observation:#1}}
\newcommand{\obsref}[1]{\HLink{Observation}{observation:#1}}

\newcommand{\seclab}[1]{\label{sec:#1}}
\newcommand{\secref}[1]{\HLink{Section}{sec:#1}}

\newcommand{\itemlab}[1]{\label{item:#1}}
\newcommand{\itemref}[1]{\HLinkSuffix{}{item:#1}{}}

\newcommand{\lemlab}[1]{\label{lemma:#1}}
\newcommand{\lemref}[1]{\HLink{Lemma}{lemma:#1}}%

\providecommand{\deflab}[1]{\label{def:#1}}
\newcommand{\defref}[1]{\HLink{Definition}{def:#1}}

\providecommand{\eqlab}[1]{}%
\renewcommand{\eqlab}[1]{\label{equation:#1}}

\providecommand{\remove}[1]{}%
\newcommand{\Set}[2]{\left\{ #1 \;\middle\vert\; #2 \right\}}
\newcommand{\pth}[2][\!]{\mleft({#2}\mright)}%

\newcommand{\ceil}[1]{\left\lceil {#1} \right\rceil}
\newcommand{\floor}[1]{\left\lfloor {#1} \right\rfloor}

\newcommand{\cardin}[1]{\left| {#1} \right|}%

\renewcommand{\th}{th\xspace}

\renewcommand{\Re}{\mathbb{R}}%
\usepackage[inline]{enumitem}

\newlist{compactenumA}{enumerate}{5}%
\setlist[compactenumA]{topsep=0pt,itemsep=-1ex,partopsep=1ex,parsep=1ex,%
   label=(\Alph*)}%

\newlist{compactenuma}{enumerate}{5}%
\setlist[compactenuma]{topsep=0pt,itemsep=-1ex,partopsep=1ex,parsep=1ex,%
   label=(\alph*)}%

\newlist{compactenumI}{enumerate}{5}%
\setlist[compactenumI]{topsep=0pt,itemsep=-1ex,partopsep=1ex,parsep=1ex,%
   label=(\Roman*)}%

\newlist{compactenumi}{enumerate}{5}%
\setlist[compactenumi]{topsep=0pt,itemsep=-1ex,partopsep=1ex,parsep=1ex,%
   label=(\roman*)}%

\newlist{compactenumi*}{enumerate*}{5}%
\setlist[compactenumi*]{topsep=0pt,itemsep=-1ex,partopsep=1ex,parsep=1ex,%
   label=(\roman*)}%

\newlist{compactitem}{itemize}{5}%
\setlist[compactitem]{topsep=0pt,itemsep=-1ex,partopsep=1ex,parsep=1ex,%
   label=\ensuremath{\bullet}}%

\definecolor{blue25emph}{rgb}{0, 0, 11}
\providecommand{\emphic}[2]{%
   \textcolor{blue25emph}{%
      \textbf{\emph{#1}}}%
   \index{#2}}

\providecommand{\emphi}[1]{\emphic{#1}{#1}}

\definecolor{almostblack}{rgb}{0, 0, 0.3}
\providecommand{\emphw}[1]{{\textcolor{almostblack}{\emph{#1}}}}%

\providecommand{\Mh}[1]{#1}%
\newcommand{\Term}[1]{\textsf{#1}}

\DeclareFontFamily{U}{BOONDOX-calo}{\skewchar\font=45 }
\DeclareFontShape{U}{BOONDOX-calo}{m}{n}{
  <-> s*[1.05] BOONDOX-r-calo}{}
\DeclareFontShape{U}{BOONDOX-calo}{b}{n}{
  <-> s*[1.05] BOONDOX-b-calo}{}
\DeclareMathAlphabet{\mathcalb}{U}{BOONDOX-calo}{m}{n}
\SetMathAlphabet{\mathcalb}{bold}{U}{BOONDOX-calo}{b}{n}
\DeclareMathAlphabet{\mathbcalb}{U}{BOONDOX-calo}{b}{n}

\newcommand{\XX}{\Mh{\EuScript{X}}}%
\newcommand{\metric}{\mathsf{d}}%

\newcommand{\PS}{\Mh{P}}%
\newcommand{\PSB}{\Mh{Q}}%
\newcommand{\PSC}{\Mh{S}}%

\newcommand{\CL}{\Mh{C}}

\newcommand{\dmY}[2]{\metric\pth{#1, #2}}%
\newcommand{\pp}{\Mh{p}}%
\newcommand{\pq}{\Mh{q}}%
\newcommand{\cpX}[1]{\Mh{\mathrm{c{}p}}\pth{#1}}%
\newcommand{\diamX}[1]{\mathrm{diam}\pth{#1}}%

\newcommand{\spreadC}{\Mh{\Phi}}%
\newcommand{\spreadX}[1]{\Mh{\Phi}\pth{#1}}%

\newcommand{\ball}{\Mh{\mathsf{b}}}%
\newcommand{\ballA}{\Mh{\mathsf{B}}}%

\newcommand{\WSPD}{\Term{WSPD}\xspace}

\newcommand{\sep}{\Mh{\mathcalb{s}}}%

\newcommand{\hc}{\Mh{\square}}%
\newcommand{\eps}{{\varepsilon}}%

\newcommand{\ropt}{\Mh{r_{\mathrm{opt}}}}

\newcommand{\pa}{\Mh{p}}%
\newcommand{\pb}{\Mh{q}}%
\newcommand{\normX}[1]{\left\| {#1} \right\|}
\newcommand{\BallSet}{\Mh{\mathcal{B}}}%
\providecommand{\IntRange}[1]{\mleft\llbracket #1 \mright\rrbracket}
\newcommand{\IRX}[1]{\IntRange{#1}}%

\newcommand{\npts}{\Mh{\gamma}}%
\newcommand{\roptY}[2]{\Mh{r_{\mathrm{opt}}}\pth{#1, #2}}%
\newcommand{\radiusX}[1]{\mathrm{radius}\pth{#1}}%
\newcommand{\eraX}[1]{\Mh{\mathrm{era}}\pth{#1}}%
\newcommand{\Clustering}{\Mh{\mathcal{C}}}%
\newcommand{\kk}{\Mh{\mathcalb{k}}}%

\newcommand{\Grid}{\Mh{G}}%

\numberwithin{figure}{section}%
\numberwithin{table}{section}%
\numberwithin{equation}{section}%

\BibLatexMode{\bibliography{good_pairs}}

\begin{document}

\title{On Clusters that are Separated but Large}

\InRegVer{%
   \author{Sariel Har-Peled%
      \SarielThanks{Work on this paper was partially supported by a
         NSF AF award CCF-1907400.}%
      \and%
      Joseph Rogge%
      \JoeThanks{}%
   }%
   \date{\today} }

\maketitle

\begin{abstract}
    Given a set $\PS$ of $n$ points in $\Re^d$, consider the problem
    of computing $\kk$ subsets of $\PS$ that form clusters that are
    well-separated from each other, and each of them is large
    (cardinality wise).  We provide tight upper and lower bounds, and
    corresponding algorithms, on the quality of separation, and the
    size of the clusters that can be computed, as a function of
    $n,d,\kk,\sep$, and $\spreadC$, where $\sep$ is the desired
    separation, and $\spreadC$ is the spread of the point set $\PS$.
\end{abstract}

\InESAVer{%
   \thispagestyle{empty}%
   \newpage%
   \pagenumbering{arabic} }%

\section{Introduction}

Clustering is one of the fundamental problems in data and computer
science. We consider a variant of clustering where we are interested
in computing clusters that are tight and well-separated in relation to
each other. Unlike the classical settings, we do not require the
clusters to cover all the points, and instead we want the clusters to
be as large as possible, while providing the desired separation
properties.  One can interpret our problem as a variant of clustering
with noise (or outliers) -- which is a notoriously hard problem
\cite{dhs-pc-01}.

\begin{figure}[h!]
    \begin{tabular}{c|c|c}
      \includegraphics[page=1,width=0.3\linewidth]{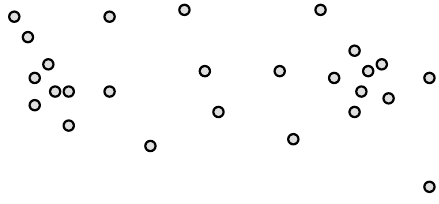}
      &
        \includegraphics[page=2,width=0.3\linewidth]{figs/well_separated}
      &
        \includegraphics[page=3,width=0.3\linewidth]{figs/well_separated}
    \end{tabular}
    \caption{A point set, two well-separated subsets (i.e., clusters),
       and the two associated balls.}
    \figlab{well:separated}
\end{figure}

\noindent%
\textbf{On different notions of separation.} %
In our settings a cluster is simply a subset of the points.  The
\emphw{size} of a cluster is the number of points in it, and when
computing a collection of clusters, the
\emphw{quality} of clustering is the cardinality of the smallest
cluster among the clusters computed.  A desired property is that
clusters are separated from each other.  Specifically, the distance
between any pair of clusters is some function of their diameters.  For
example, two sets $\CL_1, \CL_2 \subseteq \Re^d$ are \emphw{well
   $\sep$-separated} if

\begin{equation*}
    \dmY{\CL_1}{\CL_2}%
    \geq%
    \sep \max\bigl( \diamX{ \CL_1}, \diamX{\CL_2}\bigr),
\end{equation*}

\noindent%
where $\dmY{\CL_1}{\CL_2}$ is the minimum distance between points in
the two sets, and $\diamX{ \CL_i}$ is the diameter of $\CL_i$.  Here
$\sep>0$ is the \emphw{separation} parameter, and the larger it is,
the more separated the clusters are. An alternative way to view such
well-separated sets is to consider the two smallest enclosing balls of
the two clusters and require that these two balls are far from each
other, see \figref{well:separated}. If we demand that both clusters
are of large size, than there is a trade-off between the separation of
the clusters, and their quality (i.e., the minimum number of points in
either cluster).

A somewhat weaker notion is
\emphw{semi $\sep$-separation}, where we require that for the two
clusters $\CL_1, \CL_2$, we have
\begin{equation*}
    \dmY{\CL_1}{\CL_2}
    \geq%
    \sep \min\bigl( \diamX{ \CL_1}, \diamX{\CL_2} \bigr),
\end{equation*}
see \figref{semi} for an example.

If one considers more than two clusters, say $\CL_1, \ldots, \CL_\kk$,
then one can further strengthen the notion of separation, requiring
that the distance between clusters is determined by the cluster with
the largest diameter. Formally, these $\kk$ clusters are
\emphw{strongly $\sep$-separated} if
\begin{equation*}
    \forall i\neq j \qquad \dmY{\CL_i}{\CL_j} \geq \sep \cdot \max_{\ell=1}^k
    \diamX{\PS_\ell}.
\end{equation*}

\begin{figure}[b!]
    \centerline{\includegraphics{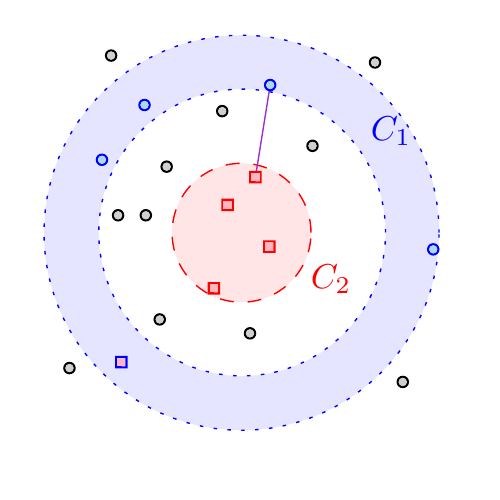}}%
    \caption{Two clusters that are semi-separated.}
    \figlab{semi}
\end{figure}

These three notions of separations, for more than two clusters, are
illustrated in \figref{separation}.

\begin{figure}[h!]
    \centerline{
       \begin{tabular}{*{1}{c}c}
         \begin{minipage}{0.4\linewidth}
             \centering \includegraphics[page=1]{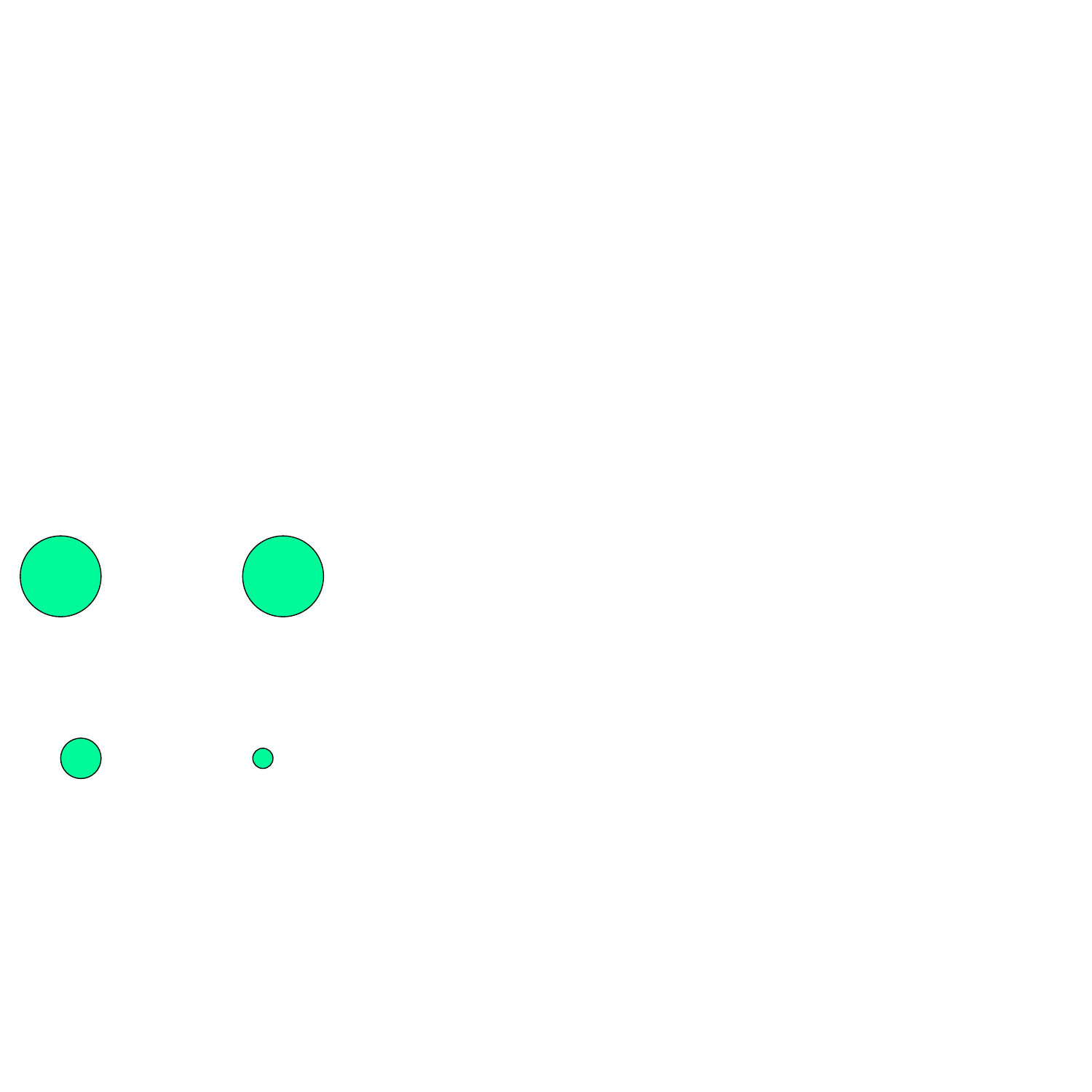}
         \end{minipage}
         &
           \qquad%
           \begin{minipage}{0.45\linewidth}
               (A) Strongly separated -- the pairwise distances
               between clusters are decided by the largest cluster. %
         \end{minipage}%
           \qquad\phantom{}%
         \\
         \hline
         \begin{minipage}{0.45\linewidth}
             \smallskip%
             \centering%
             \includegraphics[page=2]{figs/separation}%
             \smallskip%
         \end{minipage}
         &
         \begin{minipage}{0.45\linewidth}
             (B) Well separated -- the distance of separation between
             a pair of clusters is decided by the bigger cluster in
             the pair. %
         \end{minipage}
         \\
         \hline
         \begin{minipage}{0.4\linewidth}
             \centering%
             \includegraphics[page=3]{figs/separation}%
         \end{minipage}
           &
             \begin{minipage}{0.45\linewidth}
                 (C) Semi separation -- the distance between pair of
                 clusters is decided by the smaller cluster. Thus a
                 cluster might be contained ``inside'' a different
                 cluster.

             \smallskip
         \end{minipage}
       \end{tabular}%
    } \vspace{-0.7cm}%
    \caption{The different notions of separation.  }
    \figlab{separation}
\end{figure}

\paragraph{Previous work on separation in computational geometry.}

Callahan and Kosaraju \cite{ck-dmpsa-95} showed that a set $\PS$ of
$n$ points in $\Re^d$ can be decomposed into $O(\sep^d n)$ pairs, such
that all pairs of points are covered by some pair in the
decomposition, and all pairs are $\sep$-separated. This decomposition
is known as \emphw{well-separated pairs decomposition} (\WSPD). There
is also work on semi-separated pair decomposition, where one can get a
near linear bound on the total size of listing all the pairs
explicitly \cite{v-dcamc-98,ah-ncsa-12}. Both notions are widely used
in geometric approximation algorithms, as they provide compact
representation of the metric structure of the point set.

\paragraph{Separation in clustering.}

The ratio $\min_{i <j } \dmY{\CL_i}{\CL_j} / \max_{t} \diamX{\CL_t}$
is known as the \emphw{Dunn index}, and dates back to the work of Dunn
from 1973 \cite{d-fripu-73}, and is used to measure the quality of
clustering (this corresponds to the notion of strong $\sep$-separation
defined above). There are a lot of other measures of quality of
clustering, including the \emphw{Davies–Bouldin} and
\emphw{Silhouette} indices, among others. Such indices are used in
cluster analysis \cite{ells-ca-11}. Assumptions that lead to stable
clustering can be interpreted as constraints on the separation between
the ``true'' clusters \cite{l-cso-10}.

\paragraph{The task at hand.}
Here, we investigate the trade-off between the quality of separation
(under the different notions) and clusters quality/size, and provide
bounds quantifying it and algorithms for computing such good
clusterings.

\paragraph{Our results.}
In the following $\PS$ is a set of $n$ points in $\Re^d$, $\sep$ is
the desired separation, $\kk$ is the number of clusters, and
$\spreadC$ is the spread of $\PS$.  We remind the reader that for
$\kk$ clusters $\CL_1. \ldots, \CL_\kk \subseteq \PS$, the quality of
clustering is the size of the smallest cluster, that is
$\min_i \cardin{\CL_i}$.  We start with two easy upper bounds on the
quality of clustering computed.
\begin{compactenumA}
    \MedSkip{}
    \item \itemlab{spread} \textbf{Quality must depend on the spread
       (well/strong separation).}
    In \lemref{l:b:spread}, we show that the standard exponential
    point set, implies that even for two clusters
    $\CL_1,\CL_2 \subseteq \PS$ in one dimension, the quality of the
    clustering can be at most $O(n / \log \spreadC)$.

    \MedSkip{}%
    \item \itemlab{grid} \textbf{Quality drops  with the dimension.} %
    In \lemref{l:b:spread:u}, we show that the natural grid in $\Re^d$
    implies that under any notion of $\sep$-separation, the quality
    can be at most $O(n/\sep^d)$.
\end{compactenumA}
\MedSkip{}%
In particular, in high dimension, one can not get anything useful in
the worst case:
\smallskip%
\begin{compactenumA}[resume]
    \item \textbf{Quality drops exponentially with dimension.}
    Using the Johnson-Lindenstrauss lemma, we show an almost uniform
    point-set in $O( \log n)$ dimensions, such that any $2$-separated
    clusters are useless (i.e., the smaller cluster of the two
    contains a single point).  See
    \lemref{l:b:high:dim}.
\end{compactenumA}
We next combine \itemref{spread} and \itemref{grid}, to get a more
nuanced upper bound:%
\MedSkip{}
\begin{compactenumA}[resume]
    \item \textbf{A stronger upper bound using exponential grid.}
    In \secref{u:b:exp:grid}, we show that a carefully constructed
    exponential grid, implies that the quality of any two clusters
    that are well $\sep$-separated is bounded by
    $O\pth{ n / (\sep^d \log \spreadC) }$. Intuitively, this is to be
    expected from combining the above two examples, but the details
    are somewhat involved and require care.

    This upper bound construction also leads to an improved bound for
    the case of $\kk$ clusters, showing that for such clusters to be
    well/strongly $\sep$-separated, in the worst case, the quality is
    bounded by
    \begin{math}
        O\pth{ n / ( \kk \sep^d \log \spreadC) }.
    \end{math}
    See
    \corref{exp:k}.
\end{compactenumA}
\MedSkip{}%
Next, in \secref{algorithms} we study algorithms for computing such
clusterings. %
\MedSkip{}%
\begin{compactenumA}[resume]
    \item \textbf{Algorithm for computing semi $\sep$-separated $\kk$
       clusters.}
    In \lemref{h:k:sep}, we show how to compute $\kk$ clusters that
    are semi $\sep$-separated, and the quality of the computed
    clusters is $\Omega\bigl(n / (\kk \sep^d ) \bigr)$. This matches
    the above upper bound of \lemref{l:b:spread:u}.

    \MedSkip{}%
    \item \textbf{Algorithm for computing well/strong $\sep$-separated $\kk$
       clusters.}
    In \secref{quorum}, we present an algorithm to compute $\kk$
    clusters that are strongly (and thus also well) $\sep$-separated,
    and the quality of the clusters is
    $\Omega\bigl(n / (\kk \sep^d \log \spreadC) \bigr)$. This matches
    the upper bound of \corref{exp:k} mentioned above.

    \MedSkip{}%
    \item \textbf{The colored version.} %
    We also study the colored variant, where we are given $\kk$ sets
    $\PS_1, \ldots, \PS_\kk \subseteq \Re^d$ (each with $n$ points),
    and the task at hand is to compute clusters
    $\CL_1, \ldots, \CL_\kk$ that are $\sep$-separated, and
    $\CL_i \subseteq \PS_i$, for all $i$. Fortunately, a variant of
    the uncolored algorithm works in the colored case, and yields
    the same bound for the colored semi $\sep$-separated case. See
    \lemref{h:k:sep:colored}.

    Unfortunately, for three or more colors no useful clustering is
    possible if one wants strong separation, see \secref{hopeless}.

    The only remaining case is the colored well-separate case, where
    one can compute $\kk$ clusters, each of size
    $\Omega \bigl( {n}/(\kk \sep^d \log{\spreadC}) \bigr)$, see
    \lemref{k:well:sep}. Thus, the three notions of separations have
    different behaviors.
\end{compactenumA}
\MedSkip{}%
The results are summarized in \figref{results}.

\begin{figure}[t!]
    \centering%
    \begin{tabular}{|l|c|c|c|c|c|c|}
      \hline
      Notion separation
      & Bound
      & Ref
      & Dim%
      & \# cl.
      &Sep.
      &Comment
      \\
      \hline\hline%
      strong/well
      &
        $O( n /\log \spreadC)\Bigr.$
      &
        \lemref{l:b:spread}%
      &
        one
      &
        $2$
      &
        $1$
      &
        \multirow{4}{*}{%
        \begin{minipage}{2.7cm}
            \smallskip
            Upper bounds on the size of the clusters in the worst case
            \smallskip
        \end{minipage}%
      }
      \\
      \cline{1-6}
      semi%
      &
        $O\bigl( n/(\kk \sep^d) \bigr)\Bigr.$
      &
        \lemref{l:b:spread:u}%
      &
        $d$
      &
        $\kk$
      &
        $\sep$
      &
      \\[0.002cm]
      \cline{1-6}
      strong/well
      &
        $O\pth{ n / ( \kk \sep^d \log \spreadC ) }\Bigr.$
      &
        \corref{exp:k}%
      &
      &
      &
      &
      \\
      \cline{1-6}
      strong/well/semi
      &
        $\leq n/\exp\bigl( \Omega(d) \bigr)\Bigr.$
      &
        \lemref{l:b:high:dim}
      &
        high
      &
        $2$
      &
        $1$
      &
      \\[0.1cm]
      \hline
      \hline
      Strong
      $\sep$%
      &
        $\Omega_d\bigl( \frac{n}{ \kk \sep^d \log \spreadC} \bigr)\Bigr.$
      &
        \thmref{s:sep:k:d}%
      &
        $d$
      &
        $\kk$
      &
        $\sep$
      &
        Constructive
      \\
      \hline
      Semi $\sep$-separated
      &
        $\Omega\bigl(n / (\kk \sep^d ) \bigr)\Bigr.$
      &
        \lemref{h:k:sep:colored}%
      &
      &
        $\kk$
      &
        $\sep$
      &
        Constructive\\
      \hline
    \end{tabular}
    \caption{A summary of results.}
    \figlab{results}
\end{figure}

\paragraph{Techniques used.}
For the algorithms, we use as basic building blocks two tools: (i)
fast approximation algorithm for smallest enclosing ball, and (ii)
quorum clustering.

\paragraph{Connection to Ramsey theory.}

We are addressing here a natural question of finding large subset(s)
of the data that have a good structure that is better than the one
that holds for the whole input. A classical example of such a question
is finding the largest clique in the graph (i.e., Ramsey
numbers). There is also work in finite metric spaces showing that
there is always a subset that is a ``better'' metric space
\cite{blmn-omrtp-05}.

\paragraph{Paper organization.}

We start by formally defining the different notions of separation in
\secref{defs}. We then describe, in \secref{tools}, the two main
algorithmic building blocks -- tighter ball extraction and quorum
clustering. For quorum clustering we prove a key property about their
density in \lemref{depth:limit}.  In \secref{upper} we present the
various upper bounds on the quality of clustering under the various
notions of separations.  The algorithms are presented in
\secref{algorithms}.
We conclude in \secref{conclusions}.

\section{Preliminaries}
\seclab{prelims}

\subsection{Definitions}
\seclab{defs}

\begin{definition}
    A \emphw{metric space} $(\XX, \metric)$ is a set $\XX$ equipped
    with a metric $\metric$.  For two sets $X,Y \subseteq \XX$, their
    \emphi{distance} is
    $\dmY{X}{Y}=\min_{x \in X, y \in Y} \dmY{x}{y}$.  The
    \emphi{closest pair} distance of a set of points
    $\PS \subseteq \XX$, is
    \begin{math}
        \cpX{\PS} = \min_{\pp, \pq \in \PS, \pp \neq \pq}
        \dmY{\pp}{\pq}.
    \end{math}
    The \emphi{diameter} of $\PS$ is %
    \begin{math}
        \diamX{\PS} = \max_{\pp, \pq \in \PS} \dmY{\pp}{\pq}.
    \end{math}
    The \emphi{spread} of $\PS$ is
    $\spreadX{\PS} = \diamX{\PS} / \cpX{\PS}$, which is the ratio
    between the diameter and closest pair distance.

\end{definition}

\begin{defn}
    \deflab{separation}%
    Let $\PS$ be a set of points in a metric space $(\XX, \metric)$.
    Consider $\kk$ sets $\PS_1, \ldots, \PS_\kk \subseteq \PS$, and a
    parameter $\sep > 0$. The sets $\PS_1, \ldots, \PS_\kk$ are:
    \begin{compactitem}
        \smallskip%
        \item \emphi{strongly $\sep$-separated}
        if for all distinct
        $i,j$, we have
        \begin{math}
            \dmY{\PS_i}{\PS_j} \geq \sep \cdot \max_{\ell=1}^k
            \diamX{\PS_\ell}.
        \end{math}

        \smallskip%
        \item \emphi{well $\sep$-separated} if for all distinct $i,j$,
        we have
        \begin{math}
            \dmY{\PS_i}{\PS_j} \geq \sep \cdot
            \max\bigl(\diamX{\PS_i}, \diamX{\PS_j} \bigr).
        \end{math}

        \smallskip%
        \item \emphi{semi $\sep$-separated} if for all distinct $i,j$,
        we have
        \begin{math}
            \dmY{\PS_i}{\PS_j}%
            \geq%
            \sep \cdot \min \bigl( \diamX{\PS_i}, \diamX{\PS_j}
            \bigr).%
        \end{math}
    \end{compactitem}
    \smallskip%
    The \emphi{quality} of the collection is $\min_i \cardin{\PS_i}$.
    Such a collection of sets is \emphi{useless} if
    $\cardin{\PS_i } = 1$ for some $i$ (i.e., the quality is one).
\end{defn}

\begin{observation:u}
    For $\kk=2$, strong separation and well separation are the same.
    If a pair of sets $X,Y$ is $c$-separated, then it is
    $c'$-separated for all $c' \leq c$.
\end{observation:u}

\subsection{Basic tools}
\seclab{tools}

\subsubsection{Tight ball extraction}

\begin{theorem}[\cite{hr-nplta-15}]
    \thmlab{main:m:disk}%
    Given a set $\PS$ of $n$ points in $\Re^d$ and a parameter
    $\alpha$, one can compute, in expected linear time and with high
    probability, a ball $\ball$ of radius $r$, such that
    $\ropt(\PS,\alpha) \leq r \leq 2\ropt(\PS,\alpha)$, where
    $\ropt(\PS,\alpha)$ is the minimum radius of a ball covering
    $\alpha$ points of $\PS$. Furthermore, we have that $\ball$
    contains at least $\alpha$ points of $\PS$.
\end{theorem}

\subsubsection{Quorum clustering and some properties}

Let $\PS$ be a set of $n$ points in $\Re^d$. Consider the process
that, in the $i$\th iteration, does the following:
\begin{compactenumI}
    \smallskip%
    \item Computes the a ball $\ball_i$ that contains $\geq \gamma$
    points of $\PS$, where $r_i = \radiusX{\ball_i} \leq 2\rho_i$,
    where $\rho_i = \roptY{\PS}{\gamma}$ is the radius of he smallest
    ball containing $\npts$ points of $\PS$.

    \smallskip%
    \item Select exactly $\npts$ points of $\PS \cap \ball_i$ into a
    new set $\PS_i$.

    \smallskip%
    \item $\PS \leftarrow \PS \setminus \PS_i$.
\end{compactenumI}
\MedSkip{}%
This process is repeated till $\PS$ is an empty set, and assume that
$\ball_m$ and $\PS_m$ are the last ball and set computed.  (The last
set $\PS_m$ might contain less than $\npts$ points.)

The resulting partition is known as \emphi{quorum clustering}, and it
can be computed in $O( n \log n )$ time \cite{cdhks-gqsa-05,
   hr-nplta-15} (for a constant dimension $d$).

\begin{defn}
    The \emphi{era} starting at location $i$, is the longest
    subsequence $r_i, r_{i+1}, \ldots, r_j$ such that
    \begin{equation*}
        \max( r_i, r_{i+1}, \ldots, r_{j}) \leq 4r_i.
    \end{equation*}
    Let $\eraX{i}$ denote the index ending the epoch starting at $i$
    (e.g., above we have $\eraX{i} = j$).
\end{defn}

\newcommand{\epochS}{\Mh{\mathcal{B}}}%

Starting at the first location $r_1$, this naturally partitions the
quorum clustering into epochs. Setting $f(1) = \eraX{1}$, and
$f(i) = \eraX{\bigl. f(i-1) +1}$, for $i>1$.  The \emphi{$i$\th epoch}
is the subsequence $r_{f(i-1)+1}, \ldots, r_{f(i)}$.  The sequence of
sets in the $i$\th epoch is
$\epochS_i = \{ \PS_{f(i-1)+1}, \ldots, \PS_{f(i)} \}$.

\begin{observation}
    \obslab{monotone}%
    (A) Let $\rho_i$ be the radius of the smallest ball containing
    $\gamma$ (unclustered) points in the beginning of the $i$\th
    epoch. We have that $\rho_1 \leq \rho_2 \leq \cdots \leq \rho_m$.
    Let $r$ be a radius of any ball in $\epochS_i$ containing $\gamma$
    points. By the definition
    of $\rho_i$, we have
    \begin{equation*}
        \rho_i \leq r \leq 4r_i \leq 8\rho_i.
    \end{equation*}
    Therefore, the radius of any such ball in $\epochS_i$ lies in
    $[r/8, 8r]$.

    (B) For all $i$, we have $\rho_{i+1}\geq (4r_i/2) \geq 2 \rho_i$.
\end{observation}

The above implies that the quorum clustering has at most
$O(\log \spreadC)$ epochs, and the following lemma implies that balls
that are in the same epoch are sparse -- they can not cover a point
too many times.

\begin{lemma}
    \lemlab{depth:limit}%
    Let $\epochS_i$ be the set of balls computed in the $i$\th
    epoch. Any point $\pa \in \Re^d$ is contained in at most
    $d^{O(d)}$ balls of $\epochS_i$.
\end{lemma}
\begin{proof:in:appendix}{\lemref{depth:limit}}

    Fix arbitrary $\pa \in \PS$ and let $r$ denote the radius of the
    last ball in $\BallSet_i$. By definition, and by
    \obsref{monotone}, all the balls of $\BallSet_i$ have radius in
    the range $[r/8,8r]$.

    Consider the (hyper)cube $\hc$ of side length $32r$ centered at
    $\pa$. Any ball of $\BallSet_i$ that covers $\pa$ is contained
    inside $\hc$. Consider the partition of $\hc$ into a grid of cells
    with sidelength $(r/\ceil{\smash{32\sqrt{d}})}\Bigr.$. Formally we
    partition $\hc$ into a grid of $c^d$ equal sizes cubes, where
    $c \leq 2 + { \smash{16r / (r/32\sqrt{d})} } = O(\sqrt{d})$. Every
    cell in this grid has diameter at most $r/32$, and is contained as
    such in a ball of radius $r/16$. If the number of balls in
    $\BallSet_i$ that covers $\pa$ exceeds $c^d$, then the $\hc$
    contains at least $(c^d + 1) \npts$ points in the beginning of the
    epoch. Namely, one of the gird cell contains at least $\npts$
    points at this point in time. Namely,
    $\roptY{\PS}{\npts} \leq r/16$, which implies that the first ball
    in this epoch must have radius $<r/8$. But this is impossible.
\end{proof:in:appendix}

\section{Upper bounds on quality of clustering}
\seclab{upper}

\subsection{Upper bound for the uniform case}

\begin{lemma}
    \lemlab{l:b:spread:u}%
    Let $n> 0 $ be an integer number, such that $N = n^{1/d}$ is an
    integer.  Consider the grid point set
    $\PS = \IRX{ N}^d \subseteq \Re^d$, and parameters $\kk$ and
    $\sep> 0$, where $\IRX{N} = \{1,\ldots, N\}$.  Any strong, well or
    semi $\sep$-separated $\kk$ clusters
    $\CL_1, \ldots, \CL_\kk \subseteq \PS$ have the property that
    $\min_i \cardin{\CL_i} = O( n/(\kk \sep^d) )$.
\end{lemma}

\begin{proof}
    Let $\CL_1, \ldots, \CL_\kk$ be such a clustering. Let
    $t = \ceil{ c n/(\kk \sep^d)}$ for a sufficiently large constant
    $c$, and assume, for the sake of contradiction, that
    $\cardin{\CL_i} \geq t$, for all $i$.  A grid set $S$ with
    diameter $\ell$ is contained inside a hypercube of sidelength
    $\ell+1$, and thus $|S| \leq (\ell+1)^d \leq 2^d \ell^d$
    points. We conclude that $\diamX{S} \geq |S|^{1/d} /2$.  This
    implies that
    $\diamX{\CL_i} \geq t^{1/d} /2 \geq \beta n^{1/d}/( \kk^{1/d}
    \sep)$ points (if $c$ is sufficiently large), for all $i$, where
    $\beta$ is a constant to be specified shortly.  Now, we have that
    \begin{equation*}
        \dmY{\CL_i}{\CL_j} %
        \geq%
        \sep \min( \diamX{\CL_i}, \diamX{\CL_j} )%
        \geq%
        \sep \cdot \beta n^{1/d}/( \kk^{1/d} \sep)
        =%
        \beta n^{1/d}/ \kk^{1/d}.
    \end{equation*}
    Let $\ell = \bigl(\beta/(2\sqrt{d})\bigr) n^{1/d}/ \kk^{1/d}$.
    Assign each point of the grid $\IRX{N}^d$ to the cluster closest
    to it (resolve equality in an arbitrary fashion). It is easy to
    verify that each cluster gets assigned at least $(\ell/2)^d$
    points.  This would imply that $(\ell/2)^d\kk \leq n$, which fails
    if $\beta$ is sufficiently large, as
    \begin{equation*}
        (\ell/2)^d\kk%
        =%
        \pth{
           \frac{\beta/(2\sqrt{d}) n^{1/d}}{2 \kk^{1/d}}
        }^d \kk
        \geq %
        \frac{\beta}{4^d {d^d}} n.
    \end{equation*}
\end{proof}

\subsection{Upper bound with inverse logarithmic %
   dependency on the spread}

We start by constructing an exponentially spaced point set admitting
only useless clusterings. Thus, any bound on the quality of clustering
must have (inverse) logarithmic dependency on the spread of the point
set.

\begin{lemma}
    \lemlab{l:b:spread}%
    There exists a point set $\PS$ of $n$ points in the real line,
    with spread $2^{n+1}$, such that all strongly (or well)
    $1$-separated pairs are useless. Namely, for any two clusters
    $\CL_1, \CL_2 \subseteq \PS$ that are strongly $1$-separated we
    have that $\min \bigl( \cardin{\CL_1}, \cardin{\CL_2} \bigr) = 1$.
\end{lemma}

\begin{proof}
    For $i=1,\ldots, n$, let $p_i = \sum_{j=0}^i 2^j = 2^{i+1} - 1$, and
    let $\PS = \{ p_1, \ldots, p_n\}$. The distance
    $|p_i - p_{i+1}| = |p_0 - p_i| - 1$. Consider two sets
    $B_1, B_2 \subseteq \PS$ that are $1$-separated, where, without loss
    of generality, all the points of $B_2$ are bigger than all the points
    of $B_1$.  By the $1$-separation, we have that
    $\diamX{B_2} \leq \dmY{B_1}{B_2} = \max B_1 - \min B_2$.  As such,
    if $p_i, p_{i+\Delta} \in B_2$, for some $\Delta > 0$, then we
    have
    \begin{equation*}
        \max B_1%
        \leq
        \min B_2 - \diamX{B_2}%
        \leq%
        p_i - (p_{i+\Delta} - p_i)%
        =%
        2\cdot 2^{i+1} - 2 - 2^{i+\Delta + 1} + 1%
        \leq%
        -1,
    \end{equation*}
    which is impossible. It follows that the set $B_2$ can contain
    only a single point.  Namely, this pair is useless.
\end{proof}

In light of this disappointing example, we restrict ourselves to
bounds that depends on the spread of $\PS$, when looking for
well-separation.

\begin{figure}[t!]
    \centerline{\includegraphics{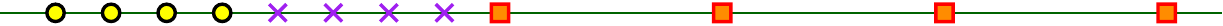}}
    \caption{No heavy triplet is possible, if there are three colors.}
    \figlab{3:colors}
\end{figure}

\subsection{Strong separation %
   is hopeless for three colors}
\seclab{hopeless}

\begin{lemma}
    There are three sets $\PS_1, \PS_2, \PS_3$ of $n$ points each on
    the real line, such that $\PS = \bigcup_i \PS_i$ has spread
    $O(n^2)$. Furthermore, for any strong $3$-separated clustering
    $\CL_1, \CL_2, \CL_3$, with $\CL_i \subseteq \PS_i$, for
    $i=1,2,3$, we have that $\min_i \cardin{\CL_i} = 1$.
\end{lemma}
\begin{proof}
    Let $\PS_1 = \IRX{n} = \{1,\ldots, n\}$,
    $\PS_2 = n + \IRX{n} = \Set{ n+ x}{x \in\IRX{n}}$, and
    \begin{equation*}
        \PS_3 = \Set{ 1 + (1+i)n}{i \in \IRX{n}}.
    \end{equation*}
    The spread of $\cup_i \PS_i$ is $O(n^2)$, See \figref{3:colors}.

    Consider any colorful strong \emph{$3$-separated} sets
    $\CL_1, \CL_2, \CL_3 $, with $\CL_i \subseteq \PS_i$, for
    $i\in \IRX{3}$.  If $\cardin{\CL_3} \geq 2$, then
    $\diamX{\CL_3} \geq n$. But this implies that $\CL_1$ and $\CL_2$
    can not be strongly $3$-separated, since
    \begin{equation*}
        3n%
        \leq
        3 \diamX{\CL_3}
        \leq%
        \dmY{ \CL_1}{\CL_2}
        \leq%
        \diamX{ \PS_1 \cup \PS_2}%
        \leq 2n,
    \end{equation*}
    which is a contradiction.
\end{proof}

\newcommand{\rankX}[1]{\Mh{\mathcalb{r}}\pth{#1}}%

\subsection{Upper bound on quality by an exponential grid}
\seclab{u:b:exp:grid}

\paragraph{Construction.}
For a point $\pp = (p_1, \ldots, p_d) \in \Re^d$, and a number
$c > 0$, let
\begin{equation*}
    \pp/c = (p_1/c, p_2/c, \ldots, p_d/c).
\end{equation*}
Similarly, for a set $X \subseteq \Re^d$, let
$X/c = \Set{\pp/c}{\pp \in X}$.

Let $n, \sep, \spreadC$ be parameters.  Let $h = \log_2 \spreadC$ --
for the sake of simplicity of exposition assume that $h$ is an
integer.  In the following, we assume that $n$ is sufficiently large
compared to $\sep$ and $h$, and $\spreadC \geq n$.

Let $R_1 = [-3,3]^d \setminus (-2,2)^d$ be a ``ring'', and pick a
maximum number $\ell$, such that the uniform grid $\Grid_1$ with
sidelength $\ell$ contains at least $n$ points in $R_1$. Let $\PS_1$
be a set of arbitrary $n$ points of $\PSB_1 = \Grid_1 \cap R_1$.  In
the following, we assume that $n$ is sufficiently large, such that
$\cardin{\Grid_1 \cap [-3,3]^d} \leq 6 n$.  For $i=2, \ldots, h$, let
$\PS_i = \PS_1 / 3^{i-1}$, $\PSB_i = \PSB_1 / 3^{i-1}$, and
$R_i = R_1 / 3^{i-1}$. Let $\PS = \cup_{i=1}^h \PS_i$. See
\figref{exp:grid}.

\begin{figure}[h!]
    \centerline{%
       \includegraphics{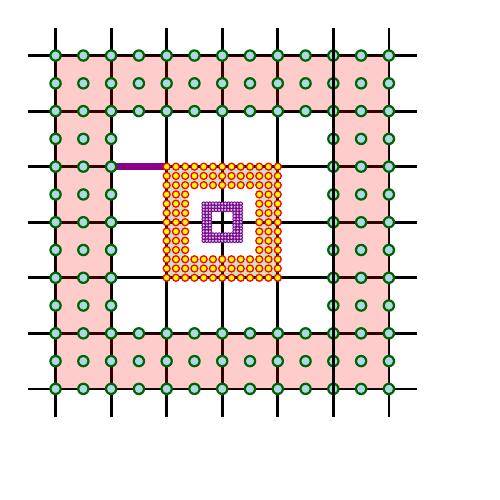}%
    }
    \caption{Exponential grid.}
    \figlab{exp:grid}
\end{figure}%

\paragraph{Analysis.}

We have that the spread of $\PS$ is bounded by
\begin{math}
    O(n 3^h) = \spreadC^{O(1)}.
\end{math}
Also, $\cardin{\PS} = \Theta(n \log \spreadC)$.

\begin{lemma}
    \lemlab{goody}%
    The \emphi{rank} of a cluster $\CL \subseteq \PS$, denoted by
    $\rankX{\CL}$, is the smallest index $j$, such that $\CL$ contains
    points of $\PS_j$. If the rank of $\CL$ is $t$, and $\CL$ contains
    at least $c n/\sep^d$ points of $\PS$, then
    $\diamX{\CL} \geq {2\diamX{R_t}}/{\sep}$, for $c$ a sufficient
    large constant that depends only on $d$.
\end{lemma}
\begin{proof}
    If $\CL$ contains points of $\cup_{\tau > t} \PS_\tau$, then by
    the ring property we have that
    \begin{equation*}
        \diamX{\CL}%
        \geq %
        \frac{\diamX{R_t}}{ 6 \sqrt{d}}%
        \geq %
        \frac{2\diamX{R_t}}{\sep}
    \end{equation*}
    (see purple segment in \figref{exp:grid}), which establishes the
    claim.

    Otherwise, $\CL \subseteq \PS_t$. By scaling, we can assume that
    $t=1$.  Let $\PSB = \Grid_1 \cap [-3,3]^d$, we have by
    construction (and the assumption that $n$ is sufficiently large),
    that $\cardin{ \PSB } \leq 6 \cardin{\PS_1} \leq 6n$.  A point set
    $\PSC \subseteq \PSB$ that contains at least
    $\gamma \cardin{\PSB} > 1$ points of $\PSB$, for some
    $\gamma \in (0,1]$, must have diameter
    $\rho \geq \gamma^{1/d} \diamX{\PSB} /(4d)$ (this follows from the
    same argument used in \lemref{l:b:spread:u}).  In particular, for
    $c$ sufficiently large, if $\CL$ contains at least $c n/\sep^d$
    points of $\PSB$, then
    \begin{equation*}
        \diamX{\CL} %
        \geq %
        \pth{\frac{c}{6\sep^d}}^{1/d} \cdot \frac{1}{4d} \cdot
        \diamX{\PSB}
        \geq %
        \frac{2}{\sep}
        \diamX{R_t},
    \end{equation*}
    for $c$ sufficiently large.
\end{proof}

\begin{lemma}
    For any two clusters $\CL_1, \CL_2 \subseteq \PS$, such that
    \begin{math}
        \dmY{\CL_1}{\CL_2}%
        \geq%
        \sep \cdot \max\bigl(\diamX{\CL_1},
        \diamX{\CL_2} \bigr),
    \end{math}
    we have that
    \begin{math}
        \displaystyle%
        \min \bigl( \cardin{\CL_1}, \cardin{\CL_2} \bigr)%
        =%
        O( n / \sep^d ).
    \end{math}
\end{lemma}

\begin{proof}
    Assume for contradiction that $\cardin{\CL_i} \geq c n /\sep^d$,
    for $i=1,2$, where $c$ is the constant specified in
    \lemref{goody}. If $r = \rankX{ \CL_1 } \leq \rankX{\CL_2}$ then
    $\CL_1, \CL_2 \subseteq H_r = [-3.3]^d/3^{r-1}$. By
    \lemref{goody}, we have that
    $\diamX{\CL} \geq {2\diamX{R_t}}/\sep = {2\diamX{H_r}}/\sep $. But
    this implies that
    \begin{equation*}
        \dmY{\CL_1}{\CL_2}%
        \geq%
        \sep \cdot \max\bigl(\diamX{\CL_1},
        \diamX{\CL_2} \bigr)
        \geq%
        2 \diamX{H_r},
    \end{equation*}
    which is a contradiction, as the two sets are contained in the
    hypercube $H_r$.

    The case that $\rankX{ \CL_1 } \geq \rankX{\CL_2}$ is handled in a
    similar fashion.
\end{proof}

This implies the following result.

\begin{theorem}
    \thmlab{exp:2}%
    Given parameter $\sep > 12d$, an integer $n$ sufficiently large,
    and a parameter $\spreadC \geq n$, one can construct a point set
    $\PS$ in $\Re^d$ of size $N = O( n \log \spreadC)$, such that
    $\spreadX{\PS} = \spreadC^{O(1)}$. Furthermore, for any two
    clusters $\CL_1, \CL_2 \subseteq \PS$ we have that if they are
    (strongly or well) $\sep$-separated then they are
    ``small''. Formally, we have
    \begin{equation*}
        \dmY{\CL_1}{\CL_2}%
        \geq%
        \sep \cdot \max\bigl(\diamX{\CL_1},
        \diamX{\CL_2} \bigr)
        \quad\implies\quad%
        \displaystyle%
        \min \bigl( \cardin{\CL_1}, \cardin{\CL_2} \bigr)%
        =%
        O\pth{ \frac{N}{ \sep^d \log \spreadC} }. %
    \end{equation*}
\end{theorem}

The above extends to $\kk$ clusters.
\begin{corollary}
    \corlab{exp:k}%
    Given parameter $\sep > 12d$, an integer $n$ sufficiently large, a
    parameter $\spreadC \geq n$, and an integer $\kk > 2$, one can
    construct a point set $\PS$ in $\Re^d$ of size
    $N = O( n \kk \log \spreadC)$, such that
    $\spreadX{\PS} = \spreadC^{O(1)}$. Furthermore, for any $\kk$
    clusters $\CL_1, \ldots, \CL_\kk \subseteq \PS$, if they are
    strongly or well $\sep$-separated, see \defref{separation}, then
    \begin{math}
        \displaystyle%
        \min_i  \cardin{\CL_i}%
        =%
        O\pth{ \frac{N}{ \kk \sep^d \log \spreadC} }.
    \end{math}
\end{corollary}

\begin{proof:in:appendix}{\corref{exp:k}}
    Let $\PS$ be the point set of \thmref{exp:2}, and let
    $\Delta= \diamX{\PS}$.  We take $\floor{\kk/2}$ copies of $\PS$,
    spacing them $\Delta$ distance from each other along a line. If
    there $\kk$ clusters, then at least two of them must belong to the
    same copy of the point set (if a cluster spans more than a single
    copy of the point set, then it is the only cluster involved in its
    point set because of the separation property). But then the bound
    of \thmref{exp:2} applies the claim immediately.
\end{proof:in:appendix}

\subsection{Upper bound in high dimensions}

\begin{theorem}[Johnson-Lindenstrauss lemma, \cite{jl-elmih-84}]
    For any $\eps \in (0,1)$, and a set of $n$ points
    $\PS \subseteq \Re^n$, there exists a linear function
    $f : \Re^n \rightarrow \Re^{d}$, for
    $d = 8\ceil{ \eps^{-2} \ln n }$, such that
    \begin{equation*}
        \forall \pa, \pb \in \PS
        \qquad
        (1 - \eps)\normX{\pa - \pb}%
        \leq%
        \normX{ f(\pa) - f(\pb)}%
        \leq%
        (1 + \eps)\normX{\pa - \pb}.
    \end{equation*}
\end{theorem}

\begin{lemma}
    \lemlab{l:b:high:dim}%
    For any $d$ sufficiently large, there exists a point set
    $\PS \subseteq \Re^d$, of size $\exp\bigl( \Omega( d) \bigr)$,
    such that all $2$-separated pairs are useless.

    In particular, there is a set of $n$ points in $\Re^{O( \log n)}$,
    such that any two clusters that are $2$-separated are useless.
\end{lemma}

\begin{proof:in:appendix}{\lemref{l:b:high:dim}}
    Let $e_1, \ldots, e_n$ be the standard orthonormal basis for
    $\Re^n$, and let $\PS = \Set{e_i/\sqrt{2}}{i=1,\ldots, n}$. By the
    Johnson-Lindenstrauss lemma, for all $0 < \eps$, there exists an
    $(1\pm \eps)$-embedding $\PS$ into $\Re^d$, where
    $d = 8 \ceil{\eps^{-2} \ln n } $.  Let $\PS' \subseteq \Re^d$ be
    this embedding of $\PS$.  Any subset $\PSB' \subseteq \PS'$, with
    two or more points, has diameter in the range
    $I = [1-\eps,1 +\eps]$. In particular, we have that for two such
    subsets $\PSB_1', \PSB_2' \subseteq \PS'$, we have that
    \begin{math}
        \dmY{\PSB_1'}{\PSB_2'}, \diamX{\PSB_1'}, \diamX{\PSB_2'} \in
        I,
    \end{math}
    Since $(1-\eps) (1+2\eps) > 1+\eps$, we conclude that for $\PS'$,
    for any $(1+2\eps)$-separated sets $\PSB_1', \PSB_2$, it must be
    that $\min( \cardin{\PSB_1'}, \cardin{\PSB_2'} ) = 1$. Namely, all
    $2$-separated pairs are useless.

    Setting $\eps = 1/2$, and stating $n$ in terms of $d$, we have
    \begin{equation*}
        \exp(d ) = \exp \pth{ 8 \ceil{\eps^{-2} \ln n }}
        \leq  e^8 n^{8/\eps^2}
        \implies%
        \exp(d/32 - 4 ) = \exp \pth{ 8 \ceil{\eps^{-2} \ln n }}
        \leq   n.
    \end{equation*}
\end{proof:in:appendix}

\begin{remark:unnumbered}
    \lemref{l:b:spread} and \lemref{l:b:high:dim} imply that any
    useful bound on the quality of well-separated pairs requires
    the bound to depend exponentially on the dimension, and
    logarithmically on the spread of the point set.
\end{remark:unnumbered}

\section{Algorithms for computing heavy and %
   separated $\kk$ clusters}
\seclab{algorithms}

\subsection{Semi-separated clusters}

\begin{lemma}
    \lemlab{h:k:sep}%
    Given a set $\PS$ of $n$ points in $\Re^d$, and parameters $\kk$
    and $\sep$, one can compute, in $O(n \kk )$ expected time, a semi
    $\sep$-separated collection of $\kk$ sets
    $\CL_1, \ldots, \CL_\kk \subseteq \PS$, such that
    $\min_i \cardin{\CL_i} = \Omega\bigl(n / (\kk \sep^d ) \bigr)$.
\end{lemma}

\begin{proof}
    Let $\alpha = \ceil{c n / (\kk \sep^d)}$, where $c$ is some
    appropriate constant. Let $\PS_0 = \PS$. In the $i$\th iteration,
    $2$-approximate the smallest ball containing $\alpha$ points of
    $\PS_{i-1}$, and let $\ball_i$ be this ball. Let $\CL_i$ be
    $\alpha$ points of $\PS_{i-1}$ contained in $\ball_i$.  Set
    $\PS_i = \PS_{i-1} \setminus \ball_i'$, where
    $\ball_i' = (2 \sep + 2) \ball_i$ is the scaling of $\ball_i$
    around its center by a factor of $2 \sep + 2$. The algorithm
    repeats this extraction step $\kk$ times.

    Observe, that $\ball_i'$ can be covered by $O( \sep^d)$ balls of
    radius $\radiusX{\ball_i}/2$, and as each such ball contains at
    most $\alpha$ points, it follow that
    $\cardin{\PS_i} \geq \cardin{\PS_{i-1}} - O(\sep^d \alpha)$. As
    such, if $c$ is chosen to be sufficiently large, we have that
    $\cardin{\PS_{\kk-1}} \geq \alpha$, which ensures the algorithm
    succeeds in extracting the $\kk$ clusters.

    By construction, for any $i < j$, we have that
    $\dmY{\CL_{i}}{\CL_j} \geq 2 \sep\radiusX{ \ball_i} \geq
    \sep\diamX{ \CL_i} $, which implies that the two sets are semi
    $\sep$-separated.
\end{proof}

We readily get the same result for the colored version.
\begin{lemma}
    \lemlab{h:k:sep:colored}%
    Given $\PS_1, \ldots, \PS_\kk \subseteq \Re^d$ be $\kk$ sets of
    $n$ points each, and a parameter $\sep$, one can compute, in
    $O(n \kk^2 )$ expected time, a semi $\sep$-separated $\kk$ sets
    $\CL_1, \ldots, \CL_\kk$, such that
    $\min_i \cardin{\CL_i} = \Omega\bigl(n / (\kk \sep^d ) \bigr)$ and
    $\CL_i \subseteq \PS_i$, for all $i$.
\end{lemma}

\begin{proof:in:appendix}{\lemref{h:k:sep:colored}}
    Initially, all the $\kk$ point sets $\PS_1, \ldots, \PS_\kk$ are
    active.

    Similar to the algorithm of \lemref{h:k:sep}, in the $i$\th
    iteration, the algorithm $2$-approximates for each \emph{active}
    set $\PS_i$ the smallest ball containing
    $\alpha = \ceil{c n / (\kk \sep^d)}$ of this set, where $c$ is
    some appropriate constant. The algorithm picks the smallest such
    ball to be $\ball_i$. For simplicity of exposition, assume
    $\ball_i$ was computed for $\PS_i$. Set $\CL_i$ to be $\alpha$
    points of $\PS_i$ that lie inside $\ball_i$, and mark $\PS_i$ as
    inactive. The algorithm now removes all the points in the active
    sets that are inside the ball $(2 + 2\sep) \ball_i$, and continues
    to the next iteration.

    The correctness of this algorithm follows the same argument is in
    \lemref{h:k:sep}.
\end{proof:in:appendix}

\subsection{Strongly separated pairs using quorum %
   clustering}
\seclab{quorum}

\begin{theorem}
    \thmlab{s:sep:k:d}%
    Given a set $\PS$ of $n$ points in $\Re^d$ with spread $\spreadC$,
    and parameters $\kk$ and $\sep$, one can compute $\kk$ disjoint
    subsets, $\CL_1, \ldots, \CL_\kk \subseteq \PS$, in $O(n \log n)$
    time, each of size at least

    \begin{equation*}
        \Omega_d\bigl( n / (\kk \sep^d \log \spreadC) \bigr),\hfill
    \end{equation*}

    such that these sets are strongly/well $\sep$-separated.
\end{theorem}
\begin{proof}
    Let $\alpha = \floor{c n / (\kk \sep^d \log \spreadC)}$, for an
    appropriate constant $c$.  We compute the quorum clustering of
    $\PS$ for the parameter $\alpha$, in $O( n \log n )$ time
    \cite{hr-nplta-15}.  By the definition of $\alpha$, there are
    $N = \ceil{n/\alpha} = \Theta( c k \sep^d \log \spreadC )$ balls
    computed by the quorum clustering.  Let $r_i$ be the radius of the
    first ball in the $i$\th epoch. The radiuses of the balls in the
    $i$\th epoch is between $[r_i/8, 8r_i]$.

    Observe that $r_{i+1} > 4r_{i}$ by construction, which readily
    implies that the number of epochs is at most $\log_2 \spreadC$. As
    such, there must be an epoch that contains at least
    $M = \Theta( c k\sep^d )$ balls. Consider a ball $\ball$ in this
    epoch, and let $\pp$ be its center and $r$ be its radius. All the
    balls in distance $\leq s\cdot 4r_i$ from it are contained in a
    ball $\ballA$ of radius $(8+2\sep)r_i$ centered at $\pp$.  By
    \lemref{depth:limit}, the number of balls of this epoch covering
    any point in $\ballA$ is bounded by $d^{O(d)}$. As such, the total
    number of balls that are not $\sep$-separated from $\ball$ is at
    most $(8+2\sep)^dd^{O(d)}$. As such, we add $\ball$ to
    $\Clustering$, and remove all the balls contained in
    $\ballA$. Clearly, we can repeat this process
    $M/(8+2\sep)^dd^{O(d)}$ times. This quantity is at least $\kk$,
    for $c$ sufficient large, as desired.

    As such, the algorithm compute the epoch with the most balls, and
    repeatedly pick a ball, add it to the clustering, and remove all
    the balls (in this epoch) that are not $\sep$-separated from the
    set of balls picked so far. After $\kk$ iterations, the resulting
    set of $\kk$ balls are all pairwise strongly $\sep$-separated, as
    desired. This later part of the algorithm can be computed in
    linear time using a grid, and we omit the straightforward details.
    The quorum clustering can be computed in $O(n \log n)$ time
    \cite{hk-drhrp-14}, the algorithm as described above is no more
    than computing an appropriate net of the centers of the balls in
    the large epoch, and this can be done in linear time
    \cite{hr-nplta-15}. Using the same techniques as
    \cite{hr-nplta-15}, it is straightforward to compute, in the
    linear time, the points $\PS$ that are contained in the selected
    balls.
\end{proof}

\subsection{Colored well-separated clustering}

\begin{lemma}
    \lemlab{k:well:sep}%
    Let $\PS_1, \ldots, \PS_\kk \subseteq \Re^d$ be $\kk$ sets of $n$
    points each, such that the spread of $\PS = \cup_i \PS_i$ is
    $\spreadC$.  For any parameter $\sep > 0$, one can compute, in
    $O(n \kk \log \spreadC)$ time, $\kk$ clusters
    $\CL_i \subseteq \PS_i$, for $i=1,\ldots, k$, such that these
    clusters are well $\sep$-separated, and
    $\min_i \cardin{\CL_i} = \Omega \bigl( {n}/(\kk \sep^d
    \log{\spreadC}) \bigr)$.
\end{lemma}
\begin{proof}
    We compute quorum clustering of each point set $\PS_i$, with sets
    being of size $\alpha = c{n}/(\kk \sep^d \log{\spreadC})$, where
    $c$ is some constant to be specified shortly. Let $\BallSet_i$ be
    the set of balls computed for the quorum clustering of $\PS_i$,
    for $i=1,\ldots, \kk$. Let $\BallSet = \cup_i \BallSet_i$. In the
    $i$\th iteration, the algorithm repeatedly pick the smallest ball
    $\ball_i$ in $\BallSet$, and assume it is a cluster of
    $\PS_i$. The algorithm add the points of $\PS_i$ covered by
    $\ball_i$ to the output, and remove all the balls of $\BallSet_i$
    from $\BallSet$. In addition, all the balls that are not well
    $\sep$-separated from $\ball_i$ in $\BallSet$ are removed from the
    set. The algorithm repeats till $\kk$ iterations are complete.

    Observe, that in an epoch of $\BallSet_j$ that involves balls that
    are of the roughly the same radius or larger than $\ball_i$,, for
    $j \geq i$, there could be at most $O(\sep^d)$ balls that are not
    well $\sep$-separated from $\ball_i$, by the packing property of
    \lemref{depth:limit}. As such, at most $O( \sep^d \log \spreadC)$
    balls are being thrown out of $\BallSet_j$, for $j >i$, in the
    $i$\th iteration. As such, for a sufficiently small $c$, the final
    set $\BallSet_\kk$ is not empty, and the algorithm indeed computes
    $\kk$ clusters each of size at least $\alpha$.

    The resulting clusters are well $\sep$-separated by construction.
\end{proof}

\section{Conclusions}
\seclab{conclusions}

We studied the problem of computing large separated clusters for a
give point set, and provided upper bounds and algorithms for computing
such ``high'' quality clusterings.%

The clustering computed can be used to seed other clustering
algorithms, such as the $\kk$-means method. We leave this as an open
problem for further research.

Another interesting open problem is to compute the best such
clustering for a given input. It is easy to show that this problem, in
general, is as hard as computing the largest clique in a graph. As
such, approximation algorithms approximating the optimal clustering
that run in polynomial time in $d, \kk$ and $\sep$ are potentially
interesting.

\InESAVer{\newpage}%

\BibLatexMode{\printbibliography}

\BibTexMode{%
   \SoCGVer{%
      \bibliographystyle{plain}%
   }%
   \NotSoCGVer{%
      \bibliographystyle{alpha}%
   }%
   \bibliography{good_pairs}
}

\InESAVer{\appendix
   \InsertAppendixOfProofs

}

\end{document}